\documentclass[aps,prr,twocolumn,superscriptaddress,
groupedaddress]{revtex4}

\usepackage{dcolumn}
\usepackage{url}
\expandafter\let\csname equation*\endcsname\relax
\expandafter\let\csname endequation*\endcsname\relax
\usepackage{amsmath,amssymb}
\usepackage{amsthm}
\usepackage{graphicx,bm,color,hyperref}
\usepackage{mathrsfs}
\usepackage{dsfont}
\usepackage{mathtools}
\usepackage{bbm,bm}
\usepackage{mathdots}
\usepackage{xcolor}
\usepackage{tabularx}
\usepackage{booktabs}
\usepackage{hyperref}
\usepackage{amssymb}
\usepackage{amsfonts}
\usepackage{changes}
\usepackage{subfigure}
\usepackage{xcolor}
\usepackage{float}

\newtheorem{theorem}{Theorem}

\newtheorem{lemma}[theorem]{Lemma}

\newcommand\dloc{q} 
\DeclareMathOperator{\hiH}{\mathcal{H}} 
\DeclareMathOperator{\C}{\mathscr{C}}
\DeclareMathOperator{\e}{\mathrm{e}}
\DeclareMathOperator{\iu}{\mathrm{i}}
\DeclareMathOperator{\Tr}{\mathrm{Tr}}

\newcommand{\ket}[1]{\vert #1 \rangle}
\newcommand{\bra}[1]{\langle #1 \vert}
\newcommand{\ketbra}[2]{\vert #1 \rangle\langle #2\vert}
\newcommand{\braket}[2]{\langle #1 \vert #2 \rangle}
\newcommand{\ceil}[1]{\left\lceil #1 \right\rceil}
\newcommand{\floor}[1]{\left\lfloor #1 \right\rfloor}
\newcommand\dH{d_H} 
\newcommand\AME{\mathrm{AME}}

\newcommand{\1}{\ensuremath{\mathbbm{1}}}
\renewcommand{\emph}{\textit}
\DeclareMathOperator{\rank}{rank}

\newcommand\ncl{n_{\text{\normalfont cl}}} 
\newcommand\nq{n_{\text{\normalfont q}}} 
\newcommand\uni {\mbox{-}\mathrm{UNI}}
\newcommand{\unimin}{\mbox{-}\mathrm{UNI}_{\min}}
\newcommand\recl{\ell}
\newcommand\req{\ell'}
\newcommand\Cl{\mathrm{Cl}} 
\newcommand\Q{\mathrm{Q}} 

\begin{document}

\title{Constructions of $k$-uniform and absolutely maximally entangled states beyond maximum distance codes}
 \author{Zahra~Raissi}
 \affiliation{ICFO-Institut de Ciencies Fotoniques, The Barcelona Institute of Science and Technology, 08860 Castelldefels (Barcelona), Spain}
 \author{Adam~Teixid\'{o}}
 \affiliation{ICFO-Institut de Ciencies Fotoniques, The Barcelona Institute of Science and Technology, 08860 Castelldefels (Barcelona), Spain}
 \author{Christian~Gogolin}
 \affiliation{ICFO-Institut de Ciencies Fotoniques, The Barcelona Institute of Science and Technology, 08860 Castelldefels (Barcelona), Spain}
 \affiliation{Institut f\"{u}r Theoretische Physik, Universit\"{a}t zu K\"{o}ln, Z\"{u}lpicher Stra\ss{}e 77, 50937 K\"{o}ln, Germany}
 \affiliation{Xanadu, 372 Richmond St W, Toronto, M5V 1X6, Canada}
 \author{Antonio~Ac\'{i}n}
 \affiliation{ICFO-Institut de Ciencies Fotoniques, The Barcelona Institute of Science and Technology, 08860 Castelldefels (Barcelona), Spain}
 \affiliation{ICREA-Instituci\'o Catalana de Recerca i Estudis Avan\c cats, Lluis Companys 23, Barcelona, 08010, Spain}

\begin{abstract}
Pure multipartite quantum states of $n$ parties and local dimension $\dloc$ are called $k$-uniform if all reductions to $k$ parties are maximally mixed. 
These states are relevant for our understanding of multipartite entanglement, quantum information protocols and the construction of quantum error correction codes. 
To our knowledge, the only known systematic construction of these quantum states is based on classical error correction codes. 
We present a systematic method to construct other examples of $k$-uniform states and show that the states derived through our construction are not equivalent to any $k$-uniform state constructed from the so-called maximum distance separable error correction codes.
Furthermore, we use our method to construct several examples of absolutely maximally entangled states whose existence was open so far.
\end{abstract}

\maketitle

\section{Introduction}
Multipartite entangled states play an important role in many quantum information processing tasks, like quantum secret sharing, quantum error correcting codes, and also in the context of high energy physics \cite{Hillery,Gottesman-thesis,Scott2004,Gottesman,Latorre,Preskill}.
All of these processes and applications depend on the property of the multipartite entangled states that are used as a resource. Providing a general framework for multipartite entanglement represents a highly complex problem, probably out of reach. Therefore, many efforts have focused on the study of relevant sets of states such as, for instance, graph states~\cite{Hein,Hein-2006} or tensor network states~\cite{Orus}.

Recently, a special class of states have attracted the attention for a wide range of tasks.
These states are called $k$-uniform states (or for simplicity $k \uni$ states), and they have the property that all of their reductions to $k$ parties are maximally mixed. An $n$-qudit state $\ket{\psi}$ in $\hiH(n,\dloc)\coloneqq \mathbb{C}_\dloc^{\otimes n}$ is $k$-uniform, and denoted in what follows by $k \uni (n,\dloc)$, whenever
\begin{equation}
\rho_S = \Tr_{S^c} \ketbra\psi\psi \propto \1 \qquad \forall S \subset \{1,\ldots,n\}, |S| \leq k \ ,
\end{equation}
where $S^c$ denotes the complementary set of $S$.
The Schmidt decomposition implies that a state can be at most $\floor{n/2} \uni$, i.e., $k\leq \floor{n/2}$. 
Operationally, in a $k \uni$ state any subset of at most $k$ parties is maximally entangled with the rest.
The $\floor{n/2} \uni$ states are called Absolutely Maximally Entangled (AME) because they are maximally
entangled along any splitting of the $n$ parties into two groups. Similarly, we denote an AME state 
in $\hiH(n,\dloc)$ by $\AME (n,\dloc)$. 

Despite their natural definition, little is known about the properties of  $k \uni$ states, such as for which
values of the tuple $(k,n,q)$ they exist or systematic methods for their construction.  
In~\cite{Dardo-Karol,Arnau,Dardo} these states were related to some classes of combinatorial designs known as orthogonal arrays (OA), and their quantum counterpart, quantum orthogonal arrays (QOA). To our knowledge, the  
most general method to construct $k \uni$ states is based on a connection between them and a family of
classical error correcting codes known as maximum distance separable (MDS)~\cite{Helwig-graphstates,ZCAA}. 
The resulting states are called of \emph{minimal support}, as they can be expressed with the minimum number of product terms
needed to guarantee that the reduced states are maximally mixed.

In this work, we introduce a systematic method of constructing $k \uni$ states. We call this method Cl+Q because it combines 
a given classical MDS code with a basis made of $k\uni$ quantum states.
We prove that our method is different from previous constructions as the derived states may not be 
of minimal support. 
In fact, we show that our states cannot be obtained from any state of minimal support by stochastic local operations and classical communication (SLOCC). 
We also use our method to construct $k \uni$ states with smaller local dimension $\dloc$ 
compared to the same $k \uni$ state constructed from MDS codes. We then show how the $k\uni$ states
derived through our construction are example of graph states and provide the corresponding graph, which 
is different from the graphs associated to states of minimal support.
Finally, we present generalizations of the Cl+Q method and use them to construct two examples of AME states whose
existence was open so far,  
namely $\AME(19,17)$ and $\AME(21,19)$.

\section{MDS codes and $k\uni$ states}
The first ingredient in our construction are classical MDS codes. 
In the language of coding theory, linear error correcting codes are usually specified by the tuple of integer numbers $[n ,k,\dH ]_\dloc$ and defined over a finite field $GF(\dloc)$.
Such codes encode $\dloc^k$ many messages specified by vectors $\vec{v}_i\in [\dloc]^k$, with $i=1,\ldots,q^k$, into a subset of codewords $\vec{c}_i\in [\dloc]^n$, all having Hamming distance $\dH$ \cite[Chapter~1]{MacWilliams}. Here $[\dloc]\coloneqq (0,\dots ,\dloc -1)$ denotes the range from $0$ to $\dloc -1$ and the Hamming distance $\dH$ between two codewords $\vec{c}_i=(c^{(i)}_1, \dots ,c^{(i)}_n)$ and $\vec{c}_j=(c^{(j)}_1, \dots ,c^{(j)}_n)$ is the number of places where they differ.
The Singleton bound~\cite{Singleton} states that for any linear code
 \begin{equation} \label{eq:hammingdistance}
  \dH \leq n -k+1\ .
 \end{equation}
A code that achieves the maximum possible minimum Hamming distance for given length and dimension is called MDS code \cite[Chapter~11]{MacWilliams}. Next, we specify MDS codes by the tuple $[n,k]_q$, as the Hamming distance follows from the saturation of the Singleton bound. 
Finally, given an $[\ncl,\recl]_q$ MDS code, it is possible to define its dual, which is an $[\ncl,\ncl-\recl]_q$ MDS code (see Appendix~\ref{Appendix-code-dual} for details on $k \unimin$ states and the number of terms they have, in expanded in the computational basis). 
In what follows, we take initial MDS codes with $\recl\leq n/2$ so that the number of codewords in the dual is $\ncl-\recl>n/2$.
 
MDS codes have been used to derive the only known systematic construction of $k\uni$  states~\cite{Scott2004,Helwig-graphstates,ZCAA}, which are also of minimal support, denoted by $k \unimin (n, \dloc)$. For a given MDS code, consider the pure quantum state corresponding to the equally weighted superposition of all the codewords $\vec{c}_i$ of the code
, i.e.,
 \begin{equation}\label{eq:kunimin}
   \ket{\psi}=
   \sum_{i=1,\ldots, q^k}\ket{\vec{c}_i} \, ,
 \end{equation}
It is instructive for what follows to see why \eqref{eq:kunimin} is a $k\uni$ state, that it, to show why all reductions up to $k$ parties are maximally mixed (more details in Appendix~\ref{Appendix-code-dual}).
For that we use two properties of MDS codes. First, since all codewords have a distance at least equal to the Singleton bound \eqref{eq:hammingdistance}, all the off-diagonal elements of the reduced density matrices of at most $k$ parties are zero. What remains to be proven is that all the diagonal elements of the reduced state of $k$ parties are equal. But this follows from the fact that any MDS code has a systematic encoder in which any set of symbols of length $k$ of the codewords can be taken as message symbols \cite[Chapter~11]{MacWilliams}, that is, all the $q^k$ possible combinations of messages appear.
Moreover, the obtained $k\uni$  states are of minimal support. This refers to the minimal number of product states needed to specify the state. For $k\uni$ states, since the reduced state of $k$ parties must be proportional to the identity, and hence of full rank, this number has to be at least equal to $\dloc ^k$, which is precisely the number of 
terms in \eqref{eq:kunimin}. 
Finally, let us recall that MDS codes over finite fields $GF(\dloc)$ have been found for the following intervals
 \begin{equation}\label{eq:existence-bound-MDS}
  \begin{cases}
    n \geq 2 \qquad \ \ \ \ \ \text {$k=1$ or $n-1$ } \\
    n \leq \dloc +2 \qquad \text{$\dloc$ is even and $k=3$ or $\dloc -1$} \\
    n \leq \dloc +1 \qquad \text{all other cases}
  \end{cases} ,
 \end{equation}
which in turn defines an existence interval of $k \unimin$ states, i.e., $k \leq \floor{n/2}$ (see \cite[Chapter~11]{MacWilliams}, \cite{Roth}).

\section{Orthonormal basis} 
The second ingredient we used in our construction are orthonormal bases where all the elements are $k\uni$ states. In principle, the $k\uni$ states in the basis can be arbitrary but in what follows we show how to construct examples of such bases starting from a $k \unimin$ state built from an $[n,k]_q$ MDS code. Let us first introduce the unitary operators $X$ and $Z$ that generalize the Pauli operators to Hilbert spaces of arbitrary dimension $\dloc \geq 2$,
\begin{align} 
  X\ket{j}&=\ket{j+1 \mod \dloc} \label{eq:defX}\\
  Z\ket{j}&=\omega^j\,\ket{j}\, , \label{eq:defZ}
\end{align}
where $\omega \coloneqq \e^{\iu 2 \pi/\dloc}$ is the $\dloc$-th root of unity. 
$X$ and $Z$ are unitary, traceless, and they satisfy the conditions $X^\dloc = Z^\dloc = \1$. We now consider operators acting on $\hiH(n,\dloc)$ consisting of tensor products of powers of these operators. In particular, we focus on the operators $M(\vec v)$ labelled by $\vec v \in [\dloc^n]$, that have the form
\begin{equation}\label{eq:Moperatordef}
  M(\vec{v}) 
   \coloneqq \underbrace{Z^{v_1}\otimes \dots \otimes Z^{v_k}}_{k} \otimes \underbrace{X^{v_{k+1}} \otimes \dots \otimes X^{v_n}}_{n-k}\ .
\end{equation}
As we see next, these $\dloc^n$ unitary operators define a basis when acting on a $k \unimin$ state.

\begin{lemma}\label{lemma:orthonormal-basis}
  Consider a $k \unimin$ state $\ket \psi \in \hiH(n,\dloc)$ and all possible vectors $\vec v_i \in [\dloc^n]$, with $i=1,\ldots,q^n$. Then, the states
$\ket{\psi_{i}} \coloneqq M(\vec{v}_i)\,\ket{\psi}$ form a complete orthonormal basis of $k \unimin$ states.
\end{lemma}

\begin{proof} 
  First, note that all the $\ket{\psi_{i}}$ are $k \uni$ states, since local unitary operations do not change the entanglement properties of the state $\ket\psi$.
  Then we should just check the orthonormality of the states, i.e., check that
  \begin{equation}
    \bra{\psi} M(\vec{v}_i)^\dagger\, M(\vec {v}_{i'}) \ket{\psi} = \prod_i \delta_{i,i'} \ .
  \end{equation}
  To show this we use the fact that, for any $k \uni$ state $\ket\psi$ constructed from an MDS code $\C=[n,k,\dH = n-k+1]_\dloc$, the Hamming distance between all the terms is at least $\dH = n-k+1$.
  The large Hamming distance between the terms in the superposition of state $\ket\psi$ implies 
  \begin{equation}
    \begin{split}
    &\bra{\psi} M(\vec{v}_i)^\dagger \, M(\vec {v}_{i'}) \ket{\psi} \\
    &= \bra{\psi} M(\vec {v}_Z^{(i)})^\dagger \, M(\vec {v}_Z^{(i')}) \ket{\psi} \prod_{i=k+1}^{n} \delta_{i,i'} \, ,
      \end{split}
  \end{equation}  
  where $M(\vec {v}_Z^{(i)})$ has the $Z$ operators of $M(\vec{v}_i)$ and no $X$ operators.
  Now, by considering the property of having $k \uni$ state, we yield
    \begin{equation}
     \begin{split}
     &\bra{\psi}  M(\vec{v}_i)^\dagger \, M(\vec{v}_{i'}) \ket{\psi} \\
     & =\Tr(M(\vec{v}_Z^{(i)})^\dagger\,M(\vec{v}_Z^{(i')}))\prod_{i=k+1}^{n} \delta_{i,i'} 
    = \prod_{i=1}^{n} \delta_{i,i'}\, .
     \end{split}
    \end{equation}
   Here we also used the fact that the operator $M(\vec{v}_Z^{(i)})^\dagger  \, M(\vec{v}_Z^{(i')})$ has weight at least $k$.
\end{proof}

In~\cite{ZCAA} this result was proven for the particular case of AME states of minimal support, leading to an AME basis.
The above lemma, generalizes the result to any $k \unimin$ states. 

\section{Constructing $k \uni$ states of non-minimal support}
We are now ready to describe our method to construct non-minimal support $k \uni (n,\dloc)$ states using the previous two ingredients.
The main idea is to combine the codeword of a given MDS code with the states of a complete $k\uni$ orthonormal basis, see figure~\ref{fig:nonminimal-support}(a).
These states are examples of QOAs, which determine a generalized quantum combinatorial designs (see \cite{Dardo} for details).

\begin{figure}
 \includegraphics[scale=0.22]{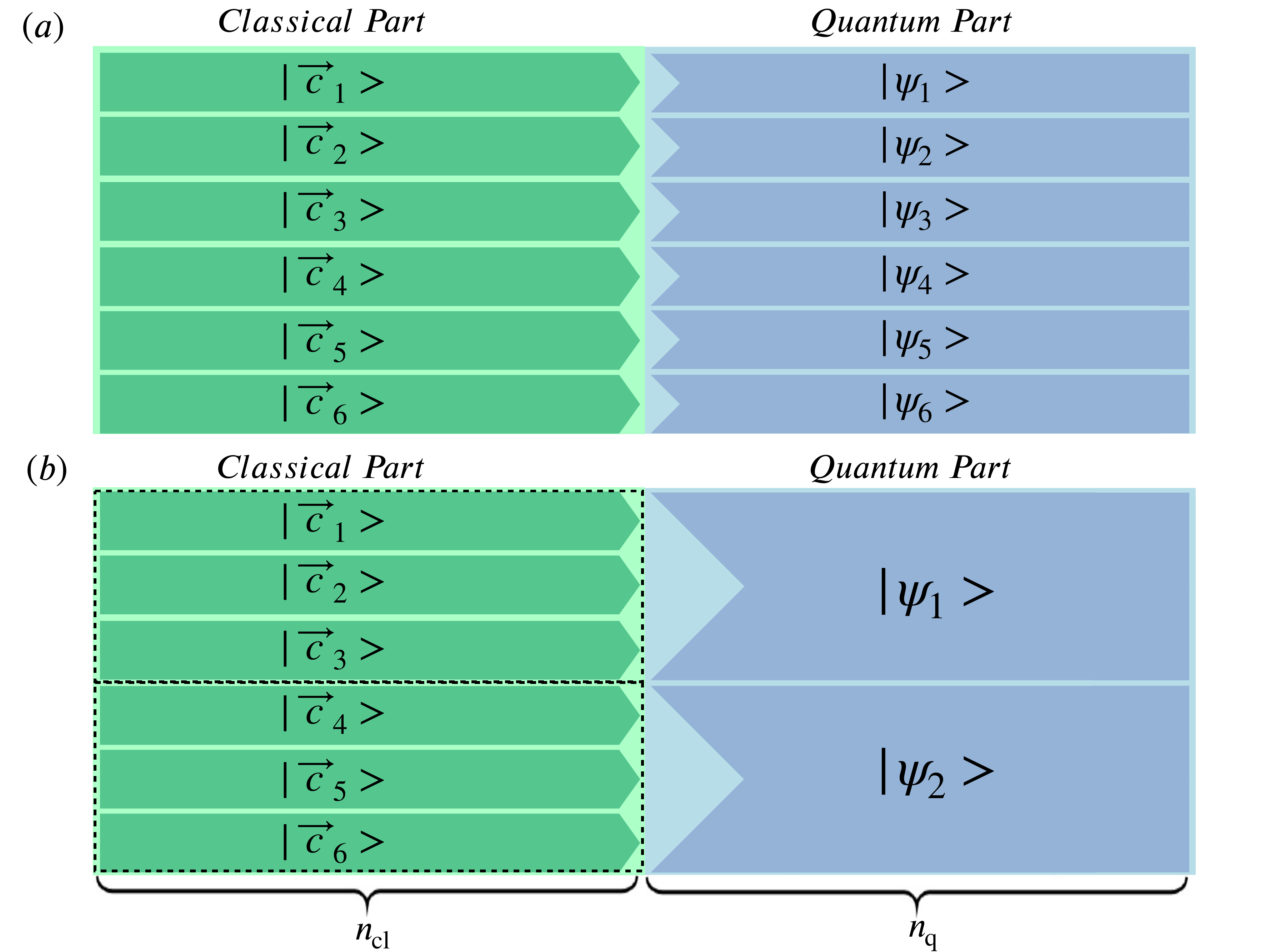}
 \centering
 \caption{\label{fig:nonminimal-support} 
  Methods of constructing $k \uni$ states.
  (a) {\it Cl+Q method.} Constructing $k \uni$ states by concatenating each codeword of an MDS code with a given $\req \uni$ state of an orthonormal basis.
  (b) {\it Cl+Q with repetition.} Constructing AME states by repeating states in the quantum part.
}
\end{figure}

 \begin{lemma}[Cl+Q method] \label{lemma:knonminimal-Cl+Q} 
   Consider an $[\ncl,\recl]_q$ MDS code of codewords $\vec{c}_i$ and a complete $\req \uni(\nq,\dloc)$ orthonormal basis with states  
   $\ket{\psi_i}$ such that $\nq = \recl$.
   Construct the state 
    \begin{equation} \label{eq:knonmin}
   \ket{\phi} =
   \sum_{i=1,\ldots,q^\recl} \underbrace{\ket{\vec{c}_i}}_{\ncl} \, \underbrace{\ket{\psi_i}}_{\nq} \ .
 \end{equation}
This state is a $(\req +1) \uni$ state of $n= \ncl + \nq$ parties.
 \end{lemma}

The condition $\nq = \recl$ is needed to ensure that the number of codewords in the code match the number of elements in the basis, as required by the construction. Note that the number of states in the $\req \uni(\nq,\dloc)$ basis is $\dloc ^{\nq}$, while the number of codewords in the MDS code is 
$\dloc^{\recl}$. 
This requirement implies that $\req<\recl$. Actually, the conditions for the lemma are slightly more general, as one can use the dual of an MDS code for the classical part. One then demands that $\nq = \ncl - \recl$ and obtains a $k=\min\{\recl +1, \req +1\} \uni$ state (for more details see Appendix~\ref{Appendix-code-dual}).

 For the purpose of the proof we need to check if the reduced density matrix 
 \begin{equation}\label{eq:reduction-sigma}
  \sigma_S = \Tr_{S^c}  \ketbra{\phi}{\phi}
  =\Tr_{S^c} ( \sum_{i,j} \ketbra{\vec{c}_i}{\vec{c}_j} \otimes \ketbra{\psi_i}{\psi_j} ) \, ,
 \end{equation}
is proportional to the identity for every set $S$ of size $|S|=k$.
In order to do so we consider the three different possibilities for $S$ when the $k$ parties are
(i)~all inside the classical part, 
(ii)~all inside the quantum part 
(iii)~split between the classical and quantum part.

 \begin{proof}[Proof of lemma \ref{lemma:knonminimal-Cl+Q}]
  First, let's consider the case (i): having a complete orthonormal basis in the quantum part ensures orthogonality, i.e., $\braket{\psi_i}{\psi_j}=\delta_{i,j}$ and therefore the off-diagonal elements of $\sigma_S$ are zero.
  In addition, and similar to what happened for the construction of $k\uni$ states from MDS codes, all the diagonal elements are equal because all possible combinations of indices appear. Therefore, $\sigma_S$ is maximally mixed.
   
  Now for the case (ii): the large Hamming distance between the terms of the classical part yields orthogonality, i.e., $\braket{\vec{c}_i}{\vec{c}_j}=\delta_{i,j}$. 
  The fact that the quantum part is a complete basis, for either choices of the classical part, implies that the reduced density matrix is a sum over all states of a basis, i.e., $\sigma_S = \sum_i \ketbra{\psi_i}{\psi_i} \propto \1_{\nq}$.
  
The case (iii) is more involved and its proof can be found in Appendix~\ref{Appendix-kuni-non}, together with more details about the construction.
 \end{proof}
In Table~\ref{tabalmaintext} we provide examples of $k\uni$ states for systems of smaller dimension than those obtained using the existing MDS codes.
\begin{table}
\begin{center}
\begin{tabular}{ |c|cccc|c| } 
 \hline
 uniform & $n$ & Cl part & Basis Q part & Cl+Q & MDS code \\
 \hline \hline
      & $n=5$ &$[3,2]_\dloc$ & Bell, $\dloc^2$ states & $\dloc \geq 2$ & $\dloc \geq 4$ \\ 
      & $n=6$ & $[4,2]_\dloc$ & Bell, $\dloc^2$ states & $\dloc \geq 3$ & $\dloc \geq 4$ \\ 
$k=2$ & $n=7$ & $[5,2]_\dloc$ & Bell, $\dloc^2$ states & $\dloc \geq 4$ & $\dloc \geq 7$ \\ 
      & $n=8$ &  $[5,3]_\dloc$ & GHZ, $\dloc^3$ states & $\dloc \geq 4$ & $\dloc \geq 7$ \\ 
      & $n=9$ & $[6,3]_\dloc$ & GHZ, $\dloc^3$ states & $\dloc \geq 4$ & $\dloc \geq 8$ \\ 
      & $n=10$ & $[7,3]_\dloc$  & GHZ, $\dloc^3$ states & $\dloc \geq 7$ & $\dloc \geq 9$ \\ 
 \hline
      & $n=11$ & $[7,4]_\dloc$ & $AME(4,\dloc)$, $\dloc^4$ states & $\dloc \geq 7$ & $\dloc \geq 11$ \\ 
      & $n=12$ & $[8,4]_\dloc$ & $AME(4,\dloc)$, $\dloc^4$ states & $\dloc \geq 7$ & $\dloc \geq 11$ \\ 
$k=3$ & $n=13$ & $[9,4]_\dloc$ & $AME(4,\dloc)$, $\dloc^4$ states & $\dloc \geq 8$ & $\dloc \geq 13$ \\ 
      & $n=14$ & $[9,5]_\dloc$ & $AME(5,\dloc)$, $\dloc^5$ states & $\dloc \geq 8$ & $\dloc \geq 13$ \\ 
      & $n=15$ & $[10,5]_\dloc$  & $AME(5,\dloc)$, $\dloc^5$ states & $\dloc \geq 9$ & $\dloc \geq 16$ \\ 
      & $n=16$ & $[11,5]_\dloc$ & $AME(5,\dloc)$, $\dloc^5$ states & $\dloc \geq 11$ & $\dloc \geq 16$ \\ 
 \hline
\end{tabular}
\end{center}
 \caption{\label{tabalmaintext} Comparison between the local dimension $\dloc$ of  different $k \unimin$ states using our construction and known MDS codes.}
\end{table}

\section{Inequivalence under SLOCC}
After presenting our construction, we now show that it provides states that could not be obtained using the 
previously known method based on MDS codes. 
In order to do so, we show that states obtained using our construction cannot be obtained by SLOCC from $k\unimin$, 
that is, they belong to different SLOCC classes. 

It is a well-known result that the number of product states needed to specify a pure state is an upper bound to 
the rank of all possible reduced states. For a $k \unimin$ state, this implies that, for any subset $ S\subset \{1,\dots ,n\}$, one has
 \begin{equation}
  \rank (\rho_S) \leq \dloc ^k \ ,
 \end{equation}
where $\rho_S=\Tr_{S^c} \ketbra{\psi}{\psi}$. It is also well known that this number cannot be increased by SLOCC~\cite{rank}.
 
Now consider $k \uni$ state $\ket \phi$ in $\hiH (n,\dloc)$ constructed from Cl+Q method.
All the reductions up to $k$ parties of the state $\ket \phi$ are maximally mixed.
However, it is possible to show that there exists at least one subset of size $|S|= k+1$ parties such that the reduced density matrix $\sigma_S =\Tr_{S^c} \ketbra{\phi}{\phi} \propto \1$.
This specific set contains $k$ parties of the classical part and one party from the quantum part.
This implies that the state $\ket \phi$ is not minimal support and hence the two states $\ket \psi$ and $\ket \phi$ cannot be mapped into the other probabilistically via LOCC. 
Therefore, they belong to different SLOCC classes.

\section{Graph states}
It is also relevant to understand the construction from the point of view of graph states. A graph $G=(V, \Gamma)$ is composed of a set $V$ of $n$ vertices and a set of weighted edges specified by the \emph{adjacency matrix} $\Gamma$ \cite{Nest,Hein,Hein-2006,Bahramgiri}, an $n \times n$ symmetric matrix such that $\Gamma_{i,j}=0$ if vertices $i$ and $j$ are not connected and $\Gamma_{i,j} >0$ otherwise. Graph states are pure quantum states specified by a graph with $\Gamma_{i,j}$. 
In this formalism, qudits are represented by the graph vertices $V$.
The graph state associated with a given graph $G$ is the $+1$ eigenstate of the following set of stabilizer operators  \cite{Nest,Hein,Hein-2006,Bahramgiri}
 \begin{equation}
   S_i = X_i \sum_j (Z_j)^{\Gamma_{i,j}} , \qquad 1 \leq i \leq n \, .
 \end{equation}
 
The $k\unimin$ states derived from MDS codes $[n,k]_\dloc$ are examples of graph states as it is possible to connect the 
adjacency matrix $\Gamma$ and the code parameters \cite{Helwig-graphstates,ZCAA}.
In particular, if one performs local Fourier transforms $F_i = \sum_{i,j} \omega^{i j} \ketbra{i}{j} $ on all the last $n-k$ parties of the state $\ket \psi$ in~\eqref{eq:kunimin}, the resulting state is a graph state corresponding to a \emph{complete bipartite graph}, see Figure \ref{fig:graph-states}(a). 
This graph is partitioned into two subsets, one containing $k$ vertices and the other one $n-k$ vertices.
The weights of the edges connecting the vertices in the two subsets depend on the details of the construction of the MDS code but the  structure is the same for all the states $\ket \psi$~\eqref{eq:kunimin}.
Note that, when $\dloc$ is a power of a prime, discrete Heisenberg-Weyl groups should be considered for the stabiliser formalism \cite{Heisenberg-Weyl, Dani}.

The graph state representation of the states $\ket \phi$ constructed from the Cl+Q method, Eq.~\eqref{eq:knonmin}, when the states in the basis are $k\unimin$ derived from an MDS code, is rather intuitive and shows the structure of the method: it is formed by concatenating  the two complete bipartite graphs associated to each MDS code or, equivalently, the corresponding $k\unimin$ state, as shown in Figure~\ref{fig:graph-states}(b).
All the details of these graph-state representations will be explained elsewhere.

\begin{figure}
 \includegraphics[scale=0.35]{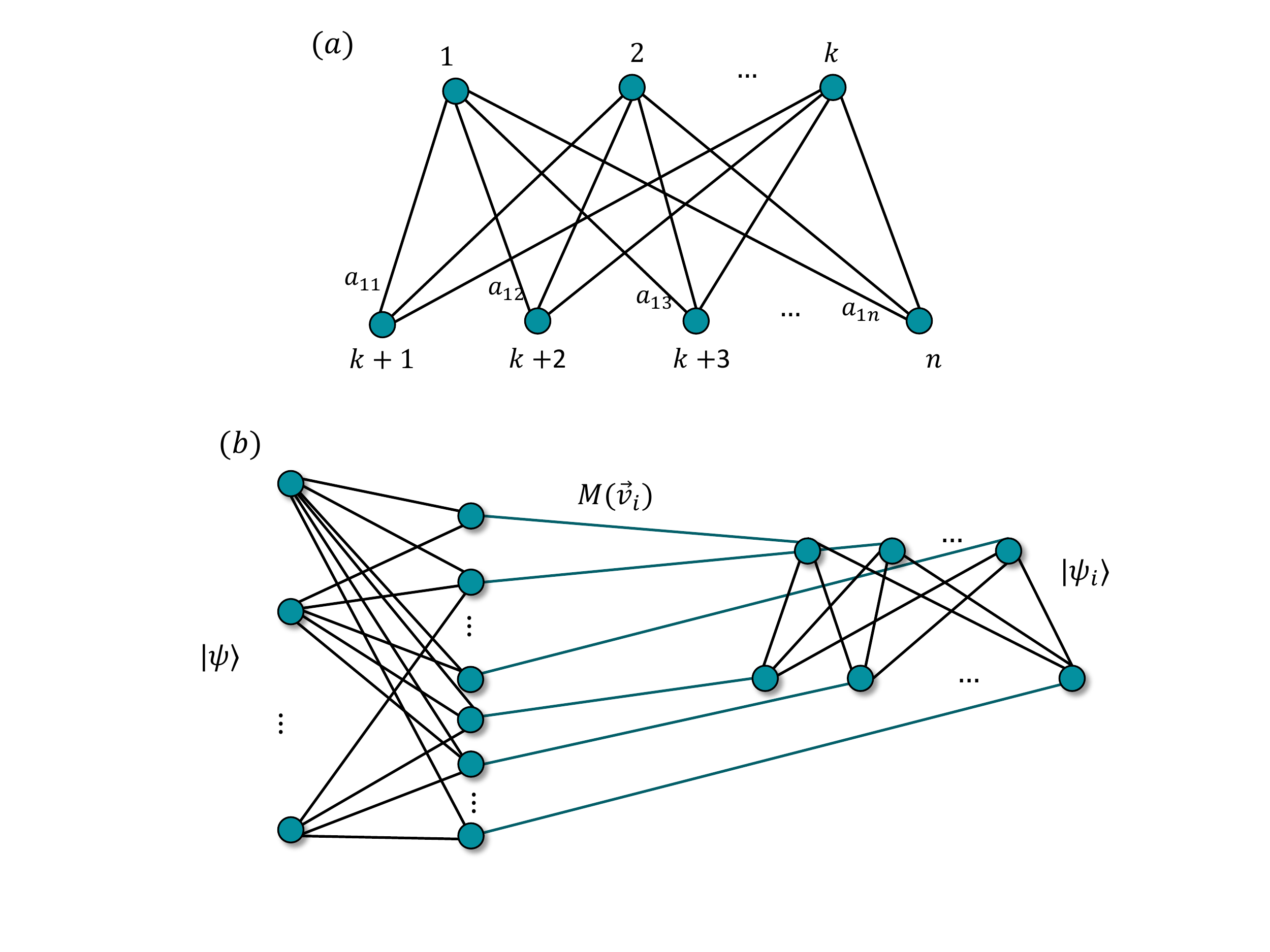}
 \centering
 \caption{\label{fig:graph-states} 
  Graph state representations of $k \uni$ states.
  (a) {\it A complete bipartite graph.} Graph state which is local unitary equivalent to the $k \unimin$ states constructed from MDS codes.
  (b) {\it  Graph state representing the $k \uni$ states constructed from the Cl+Q method.} 
  The graph can be considered as two parts connected as the method.
The left-hand side is the graph state representing the state constructed from $\ket{\psi}=\sum_i \ket{\vec{c_i}}$, i.e., the Cl part. The right-hand side is the graph state representing the Q part, states $\ket{\psi_i}$.
The operators $M(\vec{v}_i)$ describe how the two parts connect.
}
\end{figure}

\section{Constructions of unknown AME states}
We now show how using our method one can construct AME states whose existence was unknown so far.
For that we need to introduce a generalization of the method, which we call \emph{Cl+Q with repetition}, where states in the quantum part are repeated, that is, several codewords of the classical part concatenate to the same quantum state of the quantum part. 
For this to be possible, one should employ MDS codes with the property that the codewords can be distributed into subsets each forming MDS codes with smaller parameters. In particular,  we need MDS codes $\C=[\ncl,\ceil{\frac{\ncl}{2}}]_\dloc$ such that its codewords can be distributed into $\dloc^2$ subsets each forming an MDS code, with parameters $\C_i=[\ncl,\ceil{\frac{\ncl}{2}}-2]_\dloc$.
Comparing the code parameters of the MDS code $\C$ with each subclass $\C_i$, we see that they require the same number of physical qudits but the number of logical qudits decreases by $2$ (while obviously the Hamming distance increases by the same amount).
The idea is now to associate all the elements of each subclass to the same state in the Bell bases, see Figure\ref{fig:nonminimal-support}(b).

 \begin{lemma}[Cl+Q with repetition] \label{lemma:kuni-non-repetition}
  Consider an $\C= [\ncl , \ceil{\frac{\ncl}{2}}]_\dloc$ MDS code such that its codewords can be distributed into $\dloc^2$ subsets each forming MDS code with parameters $\C_i=[\ncl , \ceil{\frac{\ncl}{2}}-2]_\dloc$.
   An $AME(n,\dloc)$ state $\ket \phi$ for $n$ odd, with $n=\ncl +2$,  can be constructed by concatenating all the terms of each subclass with one of the Bell states of the quantum part, see also Figure \ref{fig:nonminimal-support}(b). 
 \end{lemma}
 In general, this configuration leads to AME states for $n$ odd when $n \leq \dloc+3$.
To show that the state $\ket \phi$ is an AME state we need to check all the reduced states $\sigma_S=\Tr_{S^c} \ketbra{\phi}{\phi}$ on up to half of the systems.
For the purpose of the proof,  we proceed as above and check three different cases, depending on how the $k$ parties are distributed between the classical and quantum part. We then use two properties of the construction:, 
(i) the fact that subsets $\C_i$ of the MDS code are also MDS codes and (ii)
the large Hamming distance between codewords of two different subsets $\C_i$ and $\C_j$, see Appendix~\ref{Appendix-kuni-non-repetition} for more details. 

What remains to be shown is that the construction can find an application, that is, that there exist MDS codes that can be distributed into $q^2$ subsets forming MDS codes.
We proved this for MDS codes with parameters $\C = [\ncl , \ceil{\ncl/2}]_\dloc$ where $\ncl \leq \dloc$, whose codewords can be distributed into $q^2$ MDS codes $\C_i =[\ncl , \ceil{\ncl/2}-2]_\dloc$, that the technique is presented in Appendix~\ref{Appendix-MDS-subsets}.
This result then allows us to construct $\AME(n\leq \dloc+2, \dloc )$ states, while $\dloc$ is an odd prime power. 
To our understanding, in some cases, like $\AME (19, 17)$ and $\AME (21, 19)$, the states were not known.
For the simplest case $\dloc =4$ we also provide a closed
form of states $\AME(7, 4)$ \cite{Markus-emails} (details of construction can be found in Appendix~\ref{Appendix-MDS-subsets}).
Table of known $\AME(n,\dloc)$ states for different local dimension $\dloc$ can be found at \cite{Markus-Roetteler,AMEtable,Felix}.

Before concluding this part, we would like to mention that the Cl+Q method can be generalized in a different way where the same 
quantum part is concatenated several times with the classical part. With this method, if $r$ is the number of times that each state of the quantum part concatenates to the terms of the classical part, the $k \uni$ state contains $n =\ncl + r \,  \nq$ many parties. This generalization will be discussed elsewhere.

\section{Conclusion}
We have presented a method that combines a classical error correcting  code with a basis of $k\uni$ states to generate other $k \uni$ states. We have shown that our construction is different from the other systematic construction previously known based on MDS codes: they belong to different SLOCC classes and have different graph-state representations. 
Then, we have used our method  to construct $k \uni$ states of $n$ parties with smaller local dimensions $\dloc$ compared to MDS codes,
and examples of AME states with its closed expression, such as $\AME(19,17)$, $\AME(21,19)$ and $\AME(7,4)$, that were unknown so far. Due to the importance that $k \uni$ and AME states have, it is an interesting avenue to explore how to use the method presented here for quantum information tasks and, in particular, in the context of quantum error correction.

\acknowledgements

We would thank Flavio Baccari, Adam Burchardt, Felix Huber, Barbara Kraus, Swapan Rana, Arnau Riera and Karol \.{Z}yczkowski for discussions and especially Markus Grassl for many useful comments.
We acknowledge support from the Spanish MINECO (QIBEQI FIS2016-80773-P, Severo Ochoa SEV-2015-0522), Fundacio Cellex, Generalitat de Catalunya (SGR 1381 and CERCA Programme), Fundacio Cellex and Mir-Puig, Generalitat de Catalunya (SGR 1381 and CERCA Programme), ERC AdG CERQUTE, the AXA Chair in Quantum Information Science, Fundaci\'{o} Catalunya and La Pedrera.

\appendix

\section{Linear codes and dual codes}\label{Appendix-code-dual}
In general, an error correcting code is denoted by $\left(n ,K,\dH \right)_\dloc$, when it encodes $K$ many messages into a subset of higher dimension $[\dloc]^n$, all having Hamming distance at least $\dH$.
Linear codes are a special class of codes whose set of messages is $K=[\dloc]^k$ for some integer $k$, and the injective map from this set of messages to the $[\dloc]^n$ set of codewords is linear.
Linear codes are usually denoted as $\C = [n,k,\dH]_\dloc$, over a finite field $GF(\dloc)$ (for the reasons of using finite fields see \cite[Chapter~3]{MacWilliams}).
Codewords of a linear code are all possible combination of the rows of a matrix, called a generator matrix $G_{k\times n}$.
For a given vector $\vec{v}_i \in [\dloc]^k$ a codeword can be written as $\vec{c}_i = \vec{v}_i \, G_{k\times n}$.
A generator matrix can always be written in the standard form  
 \begin{equation}
   G_{k \times n} = [\1_k |A]\, , 
 \end{equation}
where $\1_k$ is a $k \times k$ identity matrix and $A \in GF(\dloc)^{k \times (n-k)}$. 

Maximum distance separable (MDS) codes are those linear codes that achieve maximum possible minimum Hamming distance, Eq.\eqref{eq:hammingdistance}.
A $k \unimin$ state $\ket \psi$, Eq.~\eqref{eq:kunimin}, can be constructed by taking superposition of the computational basis states corresponding to all of the codewords.
Using the previous results, this superposition reads
 \begin{equation}\label{eq:kunimin-general}
   \ket{\psi}=
   \sum_i \ket{\vec{c}_i}=
   \sum_{i} \ket{\vec{v}_i \, G_{k  \times n}}
   =\sum_{i} \ket{\vec{v}_i,\, \vec{v}_i A}\, .
 \end{equation}

Given a linear code $\C$ it is always possible to define the \emph{dual code} $\C ^{\perp}$ such that all of its codewords are orthogonal to all the codewords of the initial code $\C$ with respect to the Euclidean inner product of the finite field \cite[Chapter~5]{MacWilliams}.
The dual code $\C^{\perp}$ of any linear MDS code $\C$ is also MDS.
If $\C$ is an MDS code with parameters $\C=[n,k,\dH=n-k+1]_\dloc$, then the dual code has parameters $\C^{\perp}=[n,n-k,\dH^{\perp}=k+1]$ \cite[Chapter~1,11]{MacWilliams}.
To avoid ambiguity we denote the MDS code with message length $k\leq n/2$ by $\C$ and its dual with message length $n-k$ by $\C^{\perp}$.
As above, one can construct the two states $\ket \psi$ and $\ket {\psi^{\perp}}$ by taking the equally weighted superposition of the  codewords of $\C$ and its dual $\C^{\perp}$, respectively.
However, considering the connection between the codewords of the original code and its dual, one can check that the states $\ket \psi$ and $\ket{\psi^{\perp}}$ can be transformed one into the other by local unitary operations, more precisely by applying Fourier gates that map the $Z$-eigenbasis into the $X$-eigenbasis to each party.
Therefore, not only  $\ket \psi$, but also $\ket {\psi^{\perp}}$ is a $k \uni$ state of minimal support.

\section{Proof of Lemma \ref{lemma:knonminimal-Cl+Q} and presenting an example}\label{Appendix-kuni-non}
For the readers convenience we discuss the proof of 
Lemma \ref{lemma:knonminimal-Cl+Q} in more detail.

 \begin{proof}
  For  the classical part in our construction, it is possible to use an MDS code $\C=[\ncl , \recl]_\dloc$ or its dual $\C^{\perp}=[\ncl , \ncl -\recl]_\dloc$.
The resulting states can be written as
   \begin{equation}\label{eq:knonmin-general}
   \begin{split}
     \ket \phi &= \sum_{i} \underbrace{\ket{\vec{c}_i}}_{\ncl} \, \underbrace{\ket{\psi_i}}_{\nq} \\
     & = \sum_{i}  \ket {\vec{v}_i \, G_{k \times n}} \, \ket{\psi_i}
     =\sum_{i} \ket {\vec{v}_i , \, \vec{v}_i A} \, \ket{\psi_i}
     \ ,
     \end{split}
   \end{equation}  
 where as above we denote by $\ket{\phi}$ ($\ket{\phi^\perp}$) the state associated to code $\C$  ($\C^\perp$). 
 The above equation is the generalized form of Eq.~\eqref{eq:knonmin}.   
  The difference between $\ket \phi$ and $\ket {\phi^\perp}$ is in the generator matrix, or alternatively the $A$ matrix. 
  For the state $\ket{\phi^\perp}$ we have $\vec{v_i} \in [\dloc]^{\ncl -\recl}$.
  
 The pure states $\ket{\phi}$ or $\ket{\phi^\perp}$ are $k \uni$ states iff the reduced density matrix $\sigma_S$ of any subset of $k$ parties, $S \subseteq \{1,\dots,n\}$ with $|S|=k$, is maximally mixed.
 This subset may be
 (i) entirely contained  in the support of the classical part $\Cl = \{1,\dots,\ncl\}$; 
 (ii) entirely contained in the support of the quantum part $\Q = \{1,\dots,\nq\}$, 
 (iii) split between the two parts $\Cl \cup \Q = \{1,\dots,n\}$.
 We consider these three different cases separately.

 \begin{itemize} 
  \item[Case (i):] If the $S$ qudits of the reduced density matrix $\sigma_S$ are contained in the classical part, $S \subseteq \Cl$, the reduced density matrix resulting from tracing out all the quantum part and the complement of $S$ in $\Cl$, $S_{\Cl}^c=S^c\cap \Cl$, of the state $\ket\phi$, Eq.~\eqref{eq:knonmin-general}, is
 \begin{equation}\label{eq:reduced-traceout-quantum}
  \begin{split}
  \sigma_S & =\Tr_{S_{\Cl}^c}\Tr_{\Q} \ketbra{\phi}{\phi} \\
  & =\sum_{i,i'} (\Tr_{S_{\Cl}^c} \ketbra{\vec{v}_i, \, \vec{v}_i A}{\vec{v}_{i'}, \, \vec{v}_{i'} \, A} ) \, \braket{\psi_i}{\psi_{i'}} \\    
  & = \sum_{i} \Tr_{S_{\Cl}^c} \ketbra{\vec{v}_i, \, \vec{v}_i A}{\vec{v}_i, \, \vec{v}_i A}\ ,
  \end{split}
 \end{equation}
 which is a direct consequence of having a complete basis in the quantum part, i.e., $\braket{\psi_{i}}{\psi_{i'}}=\delta_{i, i'}$.
 In case of considering the state $\ket{\phi^\perp}$, the same procedure holds when we calculate the reduced density matrix $\sigma_S$ with the same condition for the set $S \subseteq \Cl$. We should just replace $\vec{v}_i \in[\dloc]^{\recl}$ with $\vec{v}_i \in[\dloc]^{\ncl-\recl}$. 
As argued for state \eqref{eq:kunimin}, $\sigma_S$ is proportional to the identity matrix whenever its size is equal number of free indices in the code used in the classical part, equal to $\recl$ for the state $\ket{\phi}$ and $\ncl-\recl$ for $\ket{\phi^\perp}$.

 \item[Case (ii):] If the qudits are all contained the quantum part, $S \subseteq \Q$, the reduced density matrix $\sigma_S$ resulting from tracing out all of the qudits of the classical part and the complement of $S$ in $\Q$, $S_{\Q}^c=S^c\cap \Q$, is
 \begin{equation}
  \begin{split}
  \sigma_S & =\Tr_{\Cl}\Tr_{S_{\Q}^c} \ketbra{\phi}{\phi}\\
  & =\sum_{i,i'} \braket{\vec{v}_i }{\vec{v}_{i'}} \braket{\vec{v}_i A}{\vec{v}_{i'}A}\, (\Tr_{S_{\Q}^c} \ketbra{\psi_{i}}{\psi_{i'}}) \\    
  &=\Tr_{S_{\Q}^c}\sum_{i}\ketbra{\psi_{i}}{\psi_i}\ ,
  \end{split}
 \end{equation}
 where we have used that $\braket{\vec{v}_i}{\vec{v}_{i'}}=\delta_{i,i'}$.
 The quantum part is a complete orthogonal basis, then the reduced density matrix in this case is maximally mixed for any subset $S$ fully contained in the quantum part, which may be of size at most $\nq =\recl$ or $\nq =\ncl -\recl$ depending on the MDS code used for the classical part.
 
\item[Case (iii):] 
  Finally, we consider the case where $S \cap \Cl =S_\Cl \neq S$ and $S \cap \Q=S_\Q \neq S$. 
  We then have the general formula
  \begin{equation}\label{eq:caseIIIreducedstate}
  \begin{split}
    \sigma_S & = \Tr_{S^c}\ketbra{\phi}{\phi} \\
    & = \sum_{i,i'} \Tr_{S_{\Cl}^c} (\ketbra{\vec{v}_i , \, \vec{v}_i A}{\vec{v}_{i'} , \, \vec{v}_{i'} A}) \otimes \Tr_{S_{\Q}^c}(\ketbra{\psi_i}{\psi_{i'}}) .
    \end{split}
  \end{equation}
 We start by the state $\ket\phi$ in which the MDS code used for the classical part has $\recl\leq n_\Cl/2$ and consider the case in which $|S|=\req+1$. We first show that
  \begin{equation}\label{eq:delta}
    \Tr_{S^c_\Cl} (\ketbra{\vec{v}_i , \vec{v}_i A}{\vec{v}_{i'},\vec{v}_{i'} A}) \propto  \delta_{i,i'} ,
  \end{equation}
  for all $S$ with $|S_{\Cl}| \leq \req$.   As the terms $\ket{\vec{v}_i, \vec{v}_i A}$ that make up the classical part of the state $\ket{\phi}$ are coming from an MDS code, they are all product states in, say, the computational basis. 
 Fix any $S$, with $|S_{\Cl}| \leq \req$, and let $\{\ket s\}$ be the computational basis for $S_\Cl$ and $\{\ket t\}$ be that of $S^{c}_\Cl$.
  We can then write
  \begin{align}
  \begin{split}
   &\Tr_{S^c_\Cl} (\ketbra{\vec{v}_i ,\vec{v}_i A}{\vec{v}_{i'},\vec{v}_{i'} A}) \\
  & = \sum_{s,s',t} \ketbra{s}{s'} \braket{s,t}{\vec{v}_i,\vec{v}_i A} \braket{\vec{v}_{i'},\vec{v}_{i'} A}{s',t} \ .
   \end{split}
  \end{align}
  For $\vec{v}_i \neq \vec{v}_{i'}$, the two inner products in the right hand side of the last equation can be simultaneously non-zero only if $\ket{\vec{v}_i,\vec{v}_i A}$ and $\ket{\vec{v}_{i'},\vec{v}_{i'}A}$ are identical in at least $|S^c_\Cl|$ many locations, because otherwise they cannot both be non-orthogonal to $\ket{t}$.
  But this means that their Hamming distance could not be larger than $\dH\leq\ncl - |S^c_\Cl| = |S_\Cl| \leq \req \leq n_q/2=\recl/2$.
 But, at the same time,  we know that the Hamming distance between any two $\ket{\vec{v}_i,\vec{v}_i A}$ and $\ket{\vec{v}_{i'},\vec{v}_{i'} A}$ for $\vec{v}_i \neq \vec{v}_{i'}$ is at least $\dH = \ncl - \recl + 1 \geq \recl +1$, where the inequality follows from $\ncl \geq 2\recl$.
  These were only compatible if $\recl +1 \leq \recl/2$, which is never fulfilled.  We now use~\eqref{eq:delta} into~\eqref{eq:caseIIIreducedstate} to get
    \begin{equation}
    \begin{split}
    &\sigma_S =
    \\
    &  \sum_{i} \Tr_{S_{\Cl}^c} (\ketbra{\vec{v}_i , \, \vec{v}_i A}{\vec{v}_{i} , \, \vec{v}_{i} A}) \otimes \Tr_{S_{\Q}^c}(\ketbra{\psi_i}{\psi_{i}}) .
    \end{split}
  \end{equation}
Any set $S$ of size $\req+1$ with non-zero intersection with the classical and quantum part is such that $|S_\Q|\leq\req$. Therefore, as all the states in the quantum part $\ket{\psi_i}$ are $\req\uni$ states, one has $\Tr_{S_{\Q}^c}(\ketbra{\psi_i}{\psi_{i}})\propto\1,\forall i$. We are therefore left with 
    \begin{equation}
    \sigma_S \propto  \sum_{i} \Tr_{S_{\Cl}^c} (\ketbra{\vec{v}_i , \, \vec{v}_i A}{\vec{v}_{i} , \, \vec{v}_{i} A}) \otimes\1 ,
  \end{equation}
  which is maximally mixed because $|S_\Cl|\leq\req<\recl$.

Let us finally consider the state $\ket{\phi^\perp}$ in which the classical part is constructed from the dual code $\C ^{\perp}$ and the condition $\ncl -\recl =\nq$ is necessary.  
We can now repeat the same analysis as above. To conclude that the terms in the classical part are proportional to $\delta_{i,i'}$ we need that   $\dH^\perp = \recl +1 > |S_\Cl|$, while for the traces in the quantum part to be maximally mixed it is required that $|S_\Q|\leq\req$. These two conditions can be fulfilled if $|S|=\min \{\recl +1, \req+1\}$.
\end{itemize}    
  Now, considering all the three cases, we see that Case (iii) is the most restrictive and implies that our construction leads in general to $\min\{\recl+1,\req+1\} \uni$ states, this minimum being equal to $(\req+1)\uni$ for the state $\ket\phi$.
\end{proof}

Note that the the previous proof also implies that some reduced states $\sigma_S$ in our construction are maximally mixed even for sizes $|S|>\req$.

It just remains to present instances in which the construction applies. 
Recall that the Cl+Q method, requires  an $[\ncl,\recl]_q$ MDS code and a complete $\req \uni(\nq,\dloc)$ orthonormal basis, with $\nq=\recl$ or $\nq=n-\recl$ depending on the MDS code. For the quantum basis, we can employ the direct correspondence between minimal support states and classical MDS codes.
Then, in order to find instances of the Cl+Q method, one can simply check the known conditions for the existence of MDS codes.
To show this we use that according to Eq.~\eqref{eq:existence-bound-MDS}, we should find $\max \{\ncl,\nq\}$ for given local dimension $\dloc$.
Considering this, one simply can verify that $\max\{\ncl,\nq\}=\ncl$. 
Thus the existence of MDS code with $\ncl$ parties and local dimension $\dloc$ is enough to guarantee that such a non-minimal support $k \uni$ state constructs from our method.

As a concrete example, we can consider the state $\AME(5,\dloc)$ with the following closed form expression~\cite{Dardo}
 \begin{equation}
   \ket {\phi^\perp}
   = \sum_{l,m =0}^{\dloc -1} \ket{l,m,l+m} \ket{\psi_{(l,m)}} \ ,
 \end{equation}
where the states $\psi_{(l,m)}$ define a Bell basis
 \begin{equation}
   \ket{\psi_{(l,m)}} =
    X^l \otimes Z^m \sum_{r} \ket{r,r} \ .
 \end{equation}
For the qubit case we have
 \begin{equation}
  \begin{split}
  \ket{\phi^{\perp}}&= \ket {000}\ket{\phi^+} + \ket{011}\ket{\psi^+}\\
  & + \ket{101}\ket{\phi^-} + \ket{110}\ket{\psi-} \ ,
  \end{split}
 \end{equation}
where $\ket{\phi^{\pm}}$ and $\ket{\psi^{\pm}}$ are the Bell basis of the Hilbert space of $2$ qubits.
One can easily check that  all the reduced density matrices $\sigma_S$ up to $2$ parties are maximally mixed.

\section{Proof of Lemma \ref{lemma:kuni-non-repetition}}\label{Appendix-kuni-non-repetition}
Here we discuss how to prove Lemma \ref{lemma:kuni-non-repetition}, which is at the basis of the Cl+Q method with repetition that allows us to construct other examples of AME states.

 \begin{proof}
 In the proof of the theorem, we assume the existence of MDS codes $\C=[\ncl , \ceil{\frac{\ncl}{2}} , \dH= \floor{\frac{\ncl}{2}}+1]_\dloc$ that can be divided into $q^2$ MDS codes with smaller parameters $\C_i = [\ncl , \ceil{\frac{\ncl}{2}}-2 , \dH = \floor{\frac{\ncl}{2}}+3]_\dloc$, where $i=1,\ldots,q^2$. For each code $\C_i$, codewords are presented by $\vec{c}_{i,j}$ with $j=1,\ldots,\dloc^{ \ceil{\frac{\ncl}{2}}-2 }$.
  The state
   \begin{equation}
    \ket \phi = \sum_{i } \sum_{j}
   \underbrace{\ket{\vec{c}_{i,j}}}_{\ncl}  \underbrace{\ket{\psi_{i}}}_{\nq}
   \end{equation}
  is a modification of Eq.~\eqref{eq:knonmin-general}, and it is an AME state if all the reduced density matrices $\sigma_S = \Tr_{S^c} \ketbra{\phi}{\phi}$ are proportional to identity for $|S|\leq \floor{\frac{n}{2}} = \ceil {\frac{\ncl}{2}}$.
  As in the lemma.~\ref{lemma:knonminimal-Cl+Q}, we check three different cases for any subset $S$ of this size: this may be entirely contained in the support of the classical part $\Cl=\{1,\dots , \ncl\}$, or it can be split between the classical and quantum parts, $S_\Q$ and $S_\Cl$. 
  For the last case we have two possibilities, depending on whether the support in the  the quantum part is partial, $|S_\Q| =1$, and then $|S_\Cl| = \floor{\frac{n}{2}} -1$, or or complete, having $|S_\Q| = 2$ and $|S_\Cl| = \floor{\frac{n}{2}} -2$.
  \begin{itemize} 
   \item[Case (i):]
   If the set $S$ contain entirely in the support of the classical part,  the reduced density matrix can be written as
     \begin{equation} \label{appeq:sigmacaseII-rep}
  \begin{split}
  \sigma_S  & =\Tr_{S^c_\Cl}\Tr_{\Q} \ketbra{\phi}{\phi} \\
   & =\sum_{i} \sum_{j,j'} \Tr_{S^c_\Cl} \ketbra{\vec{c}_{i,j}}{\vec{c}_{i,j'}}\,,
  \end{split}
    \end{equation}
 where we used the orthogonality of the states $\ket{\psi_i}$.
  Since the codewords with the same value of $i$ have Hamming distance $\dH \geq \ceil{\frac{\ncl}{2}}+2$, which is larger than the size of the subset $S$, the partial trace is non-zero only when $j=j'$, having
   \begin{equation}
    \sigma_S  
    =\sum_{i,j} (\Tr_{S^c_\Cl} \ketbra{\vec{c}_{i,j}}{\vec{c}_{i,j}} ) \propto \1_{\floor{n/2}} \ .
    \end{equation}    
    where we used the fact that the number of free indices of the classical part is equal to $\ceil{\frac{\ncl}{2}} = \floor{\frac{n}{2}}$.
   \item[Case (ii):]
    The subset $S$ split between two parts such that $|S_\Q| =1$ and $|S_\Cl| = \ceil{\frac{\ncl}{2}}-1 = \floor{\frac{\ncl}{2}}$.
   Then the reduced density matrix $\sigma_S$ simplifies to
   \begin{equation}
   \begin{split}
    \sigma_S & = \Tr_{S^c} \ketbra{\phi}{\phi} \\
           & = \sum_{i,i'} \sum_{j,j'} \Tr_{S^c_\Cl} (\ketbra{\vec{c}_{i,j}}{\vec{c}_{i',j'}}) \otimes \Tr_{S^c_\Q}(\ketbra{\psi_{i}}{\psi_{i'}}) .
  \end{split}
  \end{equation}
   For the classical part, since $|S_\Cl| = \floor{\frac{\ncl}{2}}$ is smaller than the Hamming distance of the code $\C$, $d_H=\floor{\frac{\ncl}{2}}+1$, only the diagonal terms give a non-zero contribution, getting 
    \begin{equation}
    \sigma_S =
    \sum_{i,j} \Tr_{S^c _\Cl} (\ketbra{\vec{c}_{i,j}}{\vec{c}_{i,j}}) \otimes \Tr_{S^c_\Q}(\ketbra{\psi_{i}}{\psi_{i}}) .
  \end{equation}
  The trace over the quantum part gives the identity, as $\ket{\psi_i}$ are all Bell states, getting
      \begin{equation}
    \sigma_S \propto
    \sum_{i,j} \Tr_{S^c _\Cl} (\ketbra{\vec{c}_{i,j}}{\vec{c}_{i,j}}) \otimes \1_q
  \end{equation}
    The remaining sum in the classical part is the same as the reduced state obtained from the superposition of the all codewords of the MDS code $\C$, i.e., $\sum_{i,j} \ket{\vec{c}_{i,j}}$, which is an AME states of $\ncl$ parties and all its reduced density matrices up to $\floor{\frac{\ncl}{2}}$ are maximally mixed.
Putting all this together, we conclude that the reduced density matrix $\sigma_S$ is also maximally mixed.
  \item[Case (iii):] 
  We consider a subset $S$ that $|S_\Q|=|\Q|=2$ and $|S_\Cl|=\ceil{\frac{\ncl}{2}}-2 = \floor{\frac{\ncl}{2}}-1$.
  We then have the following formula
  \begin{equation}\label{appeq:case3-rep}
   \begin{split}
   \sigma_S & = \Tr_{S^c} \ketbra{\phi}{\phi} \\
    & = \sum_{i,i'} \sum_{j,j'} \Tr_{S^c_\Cl} (\ketbra{\vec{c}_{i,j}}{\vec{c}_{i'j'}}) \otimes (\ketbra{\psi_{i}}{\psi_{i'}}) .
 \end{split}
  \end{equation}
  As for case (ii), the Hamming distance between the terms of the classical part, $\dH= \floor{\frac{\ncl}{2}}+1$, is larger than the size of the subset $|S_\Cl|=\floor{\frac{\ncl}{2}}-1$, therefore $\Tr_{S^c_\Cl} (\ketbra{\vec{c}_{i,j}}{\vec{c}_{i'j'}})=0$ whenever $i\neq i'$ and $j\neq j'$ and Eq.~\eqref{appeq:case3-rep} simplifies to
 \begin{equation}\label{eq:caseiiirep}
   \sigma_S
   = \sum_{i} \sum_j \Tr_{S^c_\Cl} (\ketbra{\vec{c}_{i,j}}{\vec{c}_{i,j}}) \otimes (\ketbra{\psi_{i}}{\psi_{i}}) .
  \end{equation}
  As explained, all the codewords with the same value of $i$ define MDS codes with parameters $[\ncl , \ceil{\frac{\ncl}{2}}-2 , \dH = \ceil{\frac{\ncl}{2}}+2]_\dloc$. They all give raise to $\left(\floor{\frac{\ncl}{2}}-1\right) \uni$ states, that is, all the reduced density matrices up to $\floor{\frac{\ncl}{2}}-1$ are proportional to the identity. But the sum over index $j$ in~\eqref{eq:caseiiirep} is precisely equal to one of these reduced states for the set of parties $S_\Cl$, that is
   \begin{equation}
    \sum_{j} \Tr_{S^c_\Cl} (\ketbra{\vec{c}_{i,j}}{\vec{c}_{i,j}}) \propto \1_{\floor{\frac{\ncl}{2}}-1} \ .
   \end{equation}
  Then, we get
   \begin{equation}
    \sigma_S = \sum_{i=1}^{\dloc^2} \1_{\floor{\frac{\ncl}{2}}-1} \otimes (\ketbra{\psi_{i}}{\psi_{i}}) \ .
   \end{equation}
  The quantum part is a complete orthonormal basis, therefore $\sum_{i} \ketbra{\psi_i}{\psi_i} \propto \1_2$.
 Then, the reduced density matrix in this case $\sigma_S \propto \1_{\floor{\frac{\ncl}{2}}+1}= \1_{\floor{\frac{n}{2}}}$.  
   \end{itemize} 
 \end{proof}

\section{Distribution of codewords of an MDS code into subsets forming MDS codes with smaller parameters and corresponding AME states}\label{Appendix-MDS-subsets}
Here we show how for $\ncl \leq \dloc$ it is possible to find MDS codes $\C=[\ncl , \ceil{\ncl /2}]_\dloc$ whose codewords can be distributed into $q^2$ subsets each forming MDS codes $\C_i =[\ncl , \ceil{\ncl/2}-2]_\dloc$.
Then, we present the AME states $\AME(7,4)$, $\AME(19,17)$, and $\AME(21,19)$.

First, we describe how the MDS code $\C=[\ncl , \ceil{\ncl /2}]_\dloc$, for $\ncl \leq \dloc$ can be obtained.
In order to do so, we restrict our analysis to the biggest size $\C=[\dloc , \ceil{\dloc /2}]_\dloc$, as the other codes can be constructed in the same way.
As we discussed in Section A, in general, to find an MDS code $[n,k]_\dloc$ we need to provide a suitable generator matrix $G_{k \times n} =[\1_k |A]$.
To do that, we first recall the concept of the so-called \emph{Singleton arrays}.
Singleton arrays $S_\dloc$ represent a special case of Cauchy matrices \cite[Chapter~11]{MacWilliams} and have the property that all its square sub-matrices are non-singular.
It is known that for any finite field $GF(\dloc)$, a Singleton array of size $\dloc$ can be found as (see \cite{Roth} and Table. A2 of \cite{ZCAA})
 \begin{equation}
  S_\dloc \coloneqq \begin{array}{ccccccc}
  1 & 1 & 1 & \ldots & 1 & 1 & 1\\
  1 & a_1 & a_2 & \ldots & a_{\dloc-3} & a_{\dloc-2} &  \\
  1 & a_2 & a_3  & \ldots & a_{\dloc-2} &   &  \\
  \vdots & \vdots & \vdots & \iddots &  &   &  \\
  1 & a_{\dloc-3} & a_{\dloc-2} &   &   &   &    \\
  1 & a_{\dloc-2} &  &   &   &   &   \\
  1 &   &  &   &   &   &    \\ 
 \end{array}, \label{sq}
 \end{equation}
with
\begin{equation}
  a_i \coloneqq \frac{1}{1-\gamma^i} .
\end{equation}
where $\gamma$ is an element of $GF(\dloc)$ called \emph{primitive element}. 
All the non-zero elements of $GF(\dloc)$ can be written as some integer power of $\gamma$. 
It is also known that by taking a rectangular sub-matrix $A$ of size $k \times (n-k)$ of $S_\dloc$ one can construct a suitable generator of an MDS code \cite[Chapter~11]{MacWilliams} \cite{ZCAA}.
\begin{theorem}  \label{Corollary:MDS}
 Let $G_{k \times n} =[\1_k | A_{k \times (n-k)}]$ be the generator matrix of a code $\C$ with parameters $[n,k,\dH]_\dloc$. The following statements are equivalent:
  \begin{itemize}
   \item[(i)] $\C$ is MDS.
   \item[(ii)]every square submatrix of $A$ is nonsingular.
   \item[(iii)] any $k$ column vectors of $G_{k\times n} = [\1 |A]$ are linearly independent.
   \item[(iv)] any $n-k$ column vectors of $H_{(n-k)\times n} = [-A^T |\1]$ are linearly independent.
 \end{itemize}
\end{theorem}
For $\dloc$ being an odd prime power dimension, the biggest submatrix $A$ has size $\ceil{\dloc/2} \times \ceil{\dloc/2}$.
Using this construction, the biggest generator matrix $G_{\ceil{\dloc/2} \times {(\dloc +1)}} = [\1_{\ceil{\dloc/2}} | A]$ has size $\ceil{\dloc/2} \times {(\dloc +1)}$, and equivalently the MDS code has parameters $C=[\dloc+1, \ceil{\dloc/2}]_\dloc$.

Starting with the obtained code, there are several simple modifications to produce new codes with different parameters from the old one.
One of these manipulations is called \emph{puncturing} \cite{Gottesman-thesis,Gottesman-pastingcodes,Calderbank}, where from a given linear code $[n,k,\dH]_\dloc$ one can obtain a new code $[n-1,k,\dH-1]_\dloc$ by deleting one coordinate.
Considering this, we start from an MDS code $C=[\dloc+1, \ceil{\dloc/2}]_\dloc$ and generator matrix $G_{\ceil{\dloc/2} \times {(\dloc +1)}} = [\1_{\ceil{\dloc/2}} | A]$, by puncturing we get
  \begin{equation}\label{eq:non-standard-G}
  G_{\ceil{\dloc/2} \times \dloc}  =\left[  \begin{array}{ccc|c}
 & \1_{\ceil{\dloc/2}-1} & &   \\
  & & &  A_{\ceil{\dloc/2}\times \ceil{\dloc/2}} \\
 0 & \dots & 0 &  \\
 \end{array} \right] .
 \end{equation}
This generator matrix is not in the standard form but it constructs the MDS code $\C=[\dloc , \ceil{\dloc/2}]_\dloc$.

The second step is showing that codewords of the constructed MDS code $\C=[\dloc , \ceil{\dloc/2}]_\dloc$ distribute into subsets forming MDS codes $\C_i =[\dloc, \ceil{\dloc}-2]_\dloc$.
In order to do this, we first discuss another method of constructing new codes from old codes, called \emph{shortening} \cite{Gottesman-thesis,Gottesman-pastingcodes,Calderbank}.
Following this method, starting from a code $[n,k,\dH]_\dloc$, and by taking an appropriate subcode after deleting one coordinate, a code with parameters $[n-1,k-1,\dH'\geq\dH]_\dloc$ can be constructed.
This propagation rule will be useful several times in what follows.
We will take an appropriate subcode by choosing the codewords which have all the same value in the deleted coordinate, for instance 0. Thanks to this, all the differences between codewords must be in the coordinates that we did not delete, and thus the Hamming distance cannot decrease, $\dH'\geq\dH$.

We first show the existence of a subset $\C_0$ and then we will discuss the rest of the subsets.
We define the matrix $Q$ as
  \begin{equation}
  Q_{\ceil{\dloc/2} \times 2}  \coloneqq \left[  \begin{array}{cc}
   1 & 0\\
   a_{\ceil{\dloc/2}} & 0\\
   \vdots & \vdots \\
   a_{\dloc-2} & 0 \\
   0 & 1 \\
 \end{array} \right] ,
 \end{equation}
that contains two columns, called $Q_1$ and $Q_2$.
The $\ceil{\dloc/2}-1$ elements of $Q_1$ are exactly the same as for the $\left(\ceil{\dloc/2}+1\right)$-th column of the Singleton array $S_\dloc$. 
The biggest rectangular submatrix of the singleton array $S_\dloc$ is used to construct the generator matrix $G_{\ceil{\dloc/2} \times \dloc}$, Eq.~\eqref{eq:non-standard-G}, and the $\left(\ceil{\dloc/2}+1\right)$-th column contains $\ceil{\dloc/2}-1$ many elements that we used as $Q_1$ (we added a zero for the last element).
The column $Q_2$ is the only column of the matrix $\1_{\ceil{\dloc/2}}$ that is missing in $G_{\ceil{\dloc/2}\times \dloc}$.
Now, let's consider the following matrix
 \begin{equation}
 \left[ G|Q \right] _ {\ceil{\dloc/2} \times (\dloc+2)} =
 \left[  \begin{array}{ccc|c|cc}
    & & &  & 1 & 0\\
   & \1_{\ceil{\dloc /2}-1} & & & a_{\ceil{\dloc/2}} & 0 \\
   & & & A_{\ceil{\dloc/2} \times \ceil{\dloc/2} } & \vdots & \vdots \\
    & & & &  a_{\dloc -2} & 0 \\
    0 & \ldots & 0 & & 0 & 1 \\
 \end{array} \right] .
 \end{equation}
$G$ is the generator matrix of the MDS code $\C=[\dloc , \ceil{\dloc/2}]_\dloc$. The matrix $[G|Q]$ does not define an MDS code, Theorem \ref{Corollary:MDS} can show that its parameters are $\C=[\dloc+2 , \ceil{\dloc/2}, \floor{\dloc/2}+2]_\dloc$.
Now, we repeat the shortening process two times to get the subset $\C_0$. Every time we remove one of the last  two columns of the $[G|Q]$ matrix because $G$ is the generator matrix of the code $\C$ and we are looking for a right set of its codewords to form the code $\C_0$.
After one step shortening, removing the last row and the last column of the $[G|Q]$ matrix, we get 
\begin{equation}
\tilde{G}_{(\ceil{\dloc/2}-1)\times (\dloc +1)}= \left[  \begin{array}{c|c}
 \1_{\ceil{\dloc/2}-1}
  & 
    \begin{matrix}
   1 &  1 & \dots & 1  \\
   1 & a_1 & \dots & a_{\ceil{\dloc/2}} \\
   \vdots & \vdots & \ddots & \vdots \\
   1 & a_{\ceil{\dloc/2}-2} & \dots & a_{\dloc -2} 
  \end{matrix}
  \\
 \end{array}\right] .
\end{equation}
Theorem \ref{Corollary:MDS} tells us that the above matrix is the generator matrix of an MDS code with parameters $[\dloc+1, \ceil{\dloc/2}-1]_\dloc$. To perform the shortening process for the second time we need to find the right combination of rows of the generator matrix.
To that end we define the following matrix 
 \begin{equation}
 C_{(\ceil{\dloc/2}-1) \times (\ceil{\dloc/2}-1)} \coloneqq \left[  \begin{array}{ccc|c}
   &  & & 1  \\
   & \text{{\large $\1_{\ceil{\dloc/2}-2}$}}  & & a_{\ceil{\dloc/2}}   \\
   & & & \vdots\\
  0 & \dots & 0 & a_{\dloc -2}
 \end{array}\right]
 \end{equation}
We perform the $C^{-1}$ matrix on the generator matrix $\tilde{G}$ to get the right combination of the rows of the generator matrix to do the shortening process.
We get
\begin{align} 
C^{-1}  \tilde{G} =
   \left[  \begin{array}{ccc|ccc|c}
   & & & & & & 0 \\
   & \text{{\large $\1_{\ceil{\dloc/2}-2}$}} & & & \text{{\large $C^{-1}B$}} & & \vdots \\
   & & & & & & 0 \\
   0 & \dots & 0 & & & & 1
   \end{array}\right] ,
\end{align}
where $B$ is a submatrix of $\tilde{G}$ 
 \begin{equation}
 B_{(\ceil{\dloc/2}-1) \times (\ceil{\dloc/2}+1)} \coloneqq \left[  \begin{array}{ccccc}
   0 & 1 &  1 & \dots & 1  \\
   0 & 1 & a_1 & \dots & a_{\ceil{\dloc/2}-1} \\
   \vdots & \vdots & \vdots & \ddots & \vdots \\
    0 & 1 & a_{\ceil{\dloc/2}-3} & \dots & a_{\dloc -4} \\
   1 & 1 & a_{\ceil{\dloc/2}-2} & \dots & a_{\dloc -3} 
 \end{array}\right] .
 \end{equation}
Now, the matrix $C^{-1}\tilde{G}$ is presented in a form in which the rows are in the right combination to easily perform the shortening process.
By removing the last row and the last column we get the following matrix of size$(\ceil{\dloc/2}-2) \times \dloc$
 \begin{equation}
  G_0 =
  \left[  \begin{array}{ccc|ccc}
   & & & & & \\
   & \text{{\large $\1_{\ceil{\dloc/2}-2}$}} & & & \text{\large D}  & \\
   & & & & &  
   \end{array}\right] ,
 \end{equation}
where $D_{(\ceil{\dloc/2}-2) \times (\ceil{\dloc/2}+1)} = C^{-1}B$ removing the bottom row, and $G_0$ is the generator matrix of the shortened code, $\C_0$. We performed a shortening that keeps or grows the Hamming distance. 
Since we started with a MDS code $[\dloc+1, \ceil{\dloc/2}-1]_\dloc$, thus the shortened code is an MDS code $\C_0 = [\dloc, \ceil{\dloc/2}-2]_\dloc$. It is in fact easy to check that the Singleton bound continues to saturate.
Moreover, one verifies that the generator matrix $G_0$ is a linear combination of the rows of the generator matrix $G_{\ceil{\dloc/2} \times \dloc}$, Eq.~\eqref{eq:non-standard-G}.
This implies that the codewords of MDS code $\C_0$ are a subset of the codewords of the original MDS code $\C =[\dloc, \ceil{\dloc/2}]_\dloc$. 

It remains to show that all of the codewords of the MDS code $\C$ can be distributed into subsets each forming $\C_i=[\dloc , \ceil{\dloc/2}-2]_\dloc$.
So far we were able to show that $\dloc^{\ceil{\dloc/2}-2}$ of its codewords distribute into an MDS code with parameters $\C_0=[\dloc, \ceil{\dloc/2}-2]_\dloc$.
The fact that both MDS codes $\C$ and $\C_0$ are linear codes implies the existence of the other subsets.
Each of these subsets $\C_i$ can be achieved by adding a different codeword $\vec c_i$ of code $\C$ that is not inside the code $\C_0$ to all the codewords of code $\C_0$. 
 
Now let's consider some examples. 
The Cl+Q with repetition produces two unknown AME states, $\AME(19,17)$ and $\AME(21,19)$, as well as provides a close formula for other known AME state like $\AME(7,4)$.
The state $\AME(7,4)$ can be constructed by using MDS code with parameters $[5,3,3]_4$ and showing that all the terms can be divided into $4^2$ subgroups each forming an MDS code $[5,1,5]_4$.
Thus, the following closed form expression is an $\AME(7,4)$
 \begin{equation} \label{eq:AME(7,4)}
  \ket{\phi} = \sum_{i,j,l \in GF(4)} \ket{i, j, l, i+j+l, i+xj+(1+x)l} \otimes \ket{\varphi_{\alpha \beta}} \ ,
 \end{equation}
 where $ \varphi_{\alpha \beta}$ represents one of the Bell states such that $\alpha =i+j$, $\beta =i+xl$ over finite field $GF(4) = \{0,1,x,1+x\}$ generated by $x^2 = x+1$.
The detailed description of the subcodes $[5,1,5]_4$ connected to the Bell states $\varphi_{\alpha \beta}$ are presented in Table~\ref{table:AME(7,4)}.
Note that, in order to achieve the AME state, it is important to have different Bell states for different subclasses but the pattern of the states is not important.
 
For the other two states,  $\AME(19,17)$ and $\AME(21,19)$ we can only provide the closed form expressions of the AME states $\ket \phi$ with the $G$ and $Q$ matrices 
\begin{equation} \label{eq:newAMEstates}
 \ket{\phi} = \sum_{\vec{v}\in GF(\dloc)^{\ceil{\dloc/2}}}
 \ket{\vec{v} G} \, \ket{\psi_{\vec{v}Q}} \ ,
\end{equation}
with 
 \begin{equation}
   \ket{\psi_{\vec{v}Q}} 
 = 
   X^{\vec{v}Q_1} \otimes Z^{\vec{v}Q_2} \sum_{l = 0}^{\dloc-1} \ket{l,l}\ .
 \end{equation}
 The $G$ and $Q$ matrices to construct $\AME(19,17)$ are as follows
\begin{equation}
    G = \left[\begin{array}{ccc|ccccccccc}
 & & & 1 & 1 & 1 & 1 & 1 & 1 & 1 & 1 & 1 \\
 & & & 1 & 8 & 2 & 15 & 7 & 4 & 6 & 5 & 9 \\
 & & & 1 & 2 & 15 & 7 & 4 & 6 & 5 & 9 & 13 \\
 & \text{{\large $\1_{8 \times 8}$}}  & & 1 & 15 & 7 & 4 & 6 & 5 & 9 & 13 & 12 \\
 & & & 1 & 7 & 4 & 6 & 5 & 9 & 13 & 12 & 14 \\
 & & & 1 & 4 & 6 & 5 & 9 & 13 & 12 & 14 & 11 \\
 & & & 1 & 6 & 5 & 9 & 13 & 12 & 14 & 11 & 3 \\
 & & & 1 & 5 & 9 & 13 & 12 & 14 & 11 & 3 & 16 \\
0 & \dots & 0 & 1 & 9 & 13 & 12 & 14 & 11 & 3 & 16 & 10 \\
\end{array}\right]  \ ,
\end{equation}
and 
\begin{equation} 
    Q = \left[\begin{array}{*{2}c}
1 & 0 \\
13 & 0 \\
12 & 0 \\
14 & 0 \\
11 & 0 \\
3 & 0 \\
16 & 0 \\
10 & 0 \\
0 & 1 \\
\end{array}\right] \ .
\end{equation}
To produce the state $\AME(21,19)$ the $G$ and $Q$ matrices are
\begin{equation}
    G = \left[\begin{array}{ccc|cccccccccc}
 &  &  &  1 & 1 & 1 & 1 & 1 & 1 & 1 & 1 & 1 & 1 \\
 &  &  &  1 & 18 & 6 & 8 & 5 & 11 & 3 & 16 & 7 & 10 \\
 &  &  &  1 & 6 & 8 & 5 & 11 & 3 & 16 & 7 & 10 & 13 \\
 &  &  &  1 & 8 & 5 & 11 & 3 & 16 & 7 & 10 & 13 & 4 \\
 &  \text{{\large $\1_{9 \times 9}$}}  & & 1 & 5 & 11 & 3 & 16 & 7 & 10 & 13 & 4 & 17 \\
 &  &  &  1 & 11 & 3 & 16 & 7 & 10 & 13 & 4 & 17 & 9 \\
 &  &  &  1 & 3 & 16 & 7 & 10 & 13 & 4 & 17 & 9 & 15 \\
 &  &  &  1 & 16 & 7 & 10 & 13 & 4 & 17 & 9 & 15 & 12 \\
 &  &  &  1 & 7 & 10 & 13 & 4 & 17 & 9 & 15 & 12 & 14 \\
 0 & \dots & 0 & 1 & 10 & 13 & 4 & 17 & 9 & 15 & 12 & 14 & 2 \\
\end{array}\right] \ ,
\end{equation}
and
\begin{equation}
    Q = \left[\begin{array}{*{2}c}
1 & 0 \\
13 & 0 \\
4 & 0 \\
17 & 0 \\
9 & 0 \\
15 & 0 \\
12 & 0 \\
14 & 0 \\
2 & 0 \\
0 & 1 \\
\end{array}\right] \ .
\end{equation}
Both $\AME$ states are constructed using $G$ matrices that generate MDS codes $[\dloc , \ceil{\dloc /2}]_\dloc$, for $\dloc = 17$ or $19$ respectively, whose codewords are distributed into subsets each forming MDS codes $[\dloc , \ceil{\dloc/2}-2]_\dloc$. 
We found the right combination of the MDS codes, or alternatively the $G$ and $Q$ matrices, using a Python code.

\newpage
 \begin{table}
\begin{center}
\begin{tabular}{ |cccccc|cccccc| } 
 \hline
   $\hphantom{a}$0$\hphantom{a}$ & $\hphantom{a}$0$\hphantom{a}$ & $\hphantom{a}$0$\hphantom{a}$ & $\hphantom{a}$0$\hphantom{a}$ & $\hphantom{a}$0$\hphantom{a}$ & $\hphantom{a}$ 
   & $\hphantom{a}$0$\hphantom{a}$& $\hphantom{a}$0$\hphantom{a}$& $\hphantom{a}$3$\hphantom{a}$ & $\hphantom{a}$3$\hphantom{a}$ & $\hphantom{a}$2$\hphantom{a}$& \\
   1 & 1 & 3 & 3 & 1 &
   &1 & 1 & 0 & 0 & 3 & \\
   2 & 2 & 1 & 1 & 2 & 
   &2 & 2 & 2 & 2 & 0 & \\
   3 & 3 & 2 & 2 & 3 & $\varphi_{00}$
   &3 & 3 & 1 & 1 & 1 & $\varphi_{01}$\\
    \hline
    0 & 0 & 1 & 1 & 3 & 
    &0 & 0 & 2 & 2 & 1 & \\
    1 & 1 & 2 & 2 & 2 &
    &1 & 1 & 1 & 1 & 0 &\\
    2 & 2 & 0 & 0 & 1 & 
    &2 & 2 & 3 & 3 & 3 &\\
    3 & 3 & 3 & 3 & 0 & $\varphi_{02}$
    &3 & 3 & 0 & 0 & 2 &  $\varphi_{03}$ \\
    \hline 
    1 & 0 & 3 & 2 & 3 &
    &1 & 0 & 0 & 1 & 1 & \\
    0 & 1 & 0 & 1 & 2 &
    &0 & 1 & 3 & 2 & 0 &\\
    3 & 2 & 2 & 3 & 1 &
    &3 & 2 & 1 & 0 & 3 &\\
    2 & 3 & 1 & 0 & 0 & $\varphi_{10}$ 
    &2 & 3 & 2 & 3 & 2 & $\varphi_{11}$\\
    \hline
    1 & 0 & 2 & 3 & 0 &
    &1 & 0 & 1 & 0 & 2 &\\
    0 & 1 & 1 & 0 & 1 &
    &0 & 1 & 2 & 3 & 3 & \\
    3 & 2 & 3 & 2 & 2 &
    &3 & 2 & 0 & 1 & 0 &\\
    2 & 3 & 0 & 1 & 3 & $\varphi_{12}$ 
    & 2 & 3 & 3 & 2 & 1 & $\varphi_{13}$ \\
    \hline
    2 & 0 & 1 & 3 & 1 &
    &2 & 0 & 2 & 0 & 3 &\\
    3 & 1 & 2 & 0 & 0 &
    &3 & 1 & 1 & 3 & 2 &\\
    0 & 2 & 0 & 2 & 3 &
    &0 & 2 & 3 & 1 & 1 &\\
    1 & 3 & 3 & 1 & 2 & $\varphi_{20}$ 
    &1 & 3 & 0 & 2 & 0 & $\varphi_{21}$ \\
    \hline
    2 & 0 & 0 & 2 & 2 &
    &2 & 0 & 3 & 1 & 0 &\\
    3 & 1 & 3 & 1 & 3 &
    &3 & 1 & 0 & 2 & 1 &\\
    0 & 2 & 1 & 3 & 0 &
    &0 & 2 & 2 & 0 & 2 & \\
    1 & 3 & 2 & 0 & 1 & $\varphi_{22}$ 
    &1 & 3 & 1 & 3 & 3 & $\varphi_{23}$ \\
    \hline 
    3 & 0 & 2 & 1 & 2 &
    &3 & 0 & 1 & 2 & 0 &\\
    2 & 1 & 1 & 2 & 3 &
    &2 & 1 &2 & 1 & 1 & \\
    1 &2 & 3 & 0 & 0 &
    &1 & 2 & 0 & 3 & 2 &\\
    0 & 3 & 0 & 3 & 1 & $\varphi_{03}$ 
    &0 & 3 & 3 & 0 & 3 & $\varphi_{31}$ \\
    \hline 
    3 & 0 & 3 & 0 & 1 &
    &3 & 0 & 0 & 3 & 3 &\\
    2 & 1 & 0 & 3 & 0 &
    &2 & 1 & 3 & 0 & 2 &\\
    1 & 2 & 2 & 1 & 3 &
    &1 & 2 & 1 & 2 & 1 &\\
    0 & 3 & 1 & 2 & 2 &$\varphi_{32}$ 
    &0 & 3 & 2 & 1 & 0 & $\varphi_{33}$ \\
    \hline  
\end{tabular}
\end{center}
 \caption{\label{table:AME(7,4)} Codewrods of MDS code $[5,3,3]_4$ are distributed into $\dloc^2=16$ subsets $[5,1,5]_4$.  $\AME(7,4)$, Eq.~\eqref{eq:AME(7,4)}, formed by concatenating codewords of one subset to one of the Bell states. }
\end{table}


\begin{thebibliography}{99}
\bibitem{Markus-emails}
 The quantum code $[\![7,0,4]\!]_4$ has been found by random search, and it implies the existence of $\AME(7,4)$ state. We are not aware of a reference for such a code in the literature.
\end{thebibliography}
\end{document}